\documentclass[aps,pra,twocolumn,letterpaper,reprint]{revtex4-1}
\pdfoutput=1
\usepackage[utf8]{inputenc}
\usepackage[T1]{fontenc}
\usepackage{amsmath, amsthm, amssymb, mathtools, dsfont}
\usepackage{newtxtext} %
\usepackage[smallerops]{newtxmath} %
\usepackage{graphicx}
\usepackage{bbm}
\usepackage[normalem]{ulem} 
\usepackage[dvipsnames]{xcolor}
\usepackage{enumitem} 
\usepackage{booktabs} 
\usepackage{array} 
\usepackage[export]{adjustbox}

\usepackage{epstopdf}

\usepackage[colorlinks=true,linkcolor=blue,citecolor=blue,urlcolor=blue]{hyperref}
\usepackage{geometry}
\makeatletter 
\theoremstyle{definition}
\newtheorem{definition}{Definition}
\theoremstyle{definition}
\newtheorem{lemma}{Lemma}
\theoremstyle{definition}
\newtheorem{theorem}{Theorem}
\theoremstyle{definition}
\newtheorem{proposition}{Proposition}
\theoremstyle{definition}

\usepackage{MnSymbol} 
\newcommand\trm{\textrm}

\renewcommand\H{\mathcal{H}}
\newcommand{\1}{\mathds{1}}
\newcommand{\defeq}{\coloneqq}

\newcommand{\ot}{\otimes}
\newcommand{\op}{\oplus}

\newcommand{\bra}[1]{\langle #1 \rvert}
\newcommand{\ket}[1]{\lvert #1\rangle}
\newcommand{\braket}[2]{\langle #1| #2\rangle}
\newcommand{\ketbra}[2]{\lvert {#1}\rangle \langle {#2}\rvert}

\newcommand{\proj}[1]{\ketbra{#1}{#1}}
\newcommand{\modu}[1]{\lvert #1 \rvert}

\DeclareMathOperator{\tr}{tr}

\DeclareMathOperator{\im}{im}
\DeclareMathOperator{\supp}{supp}
\DeclareMathOperator{\rank}{rank}
\DeclareMathOperator{\vecr}{vec} 
\newcommand{\ncl}{\nonumber\\}

\newcommand*{\mathtiny}{\scriptscriptstyle}
\newcommand*{\inv}{^{-1}}
\newcommand*{\ginv}{^{-}}
\newcommand*{\ad}{^{\dagger}}

\newcommand{\Sc}{\mathcal S}
\newcommand{\Qc}{\mathcal Q}
\newcommand{\mio}{{\operatorname{MIO}}}
\newcommand{\bic}{{\operatorname{BIC}}}
\newcommand{\pic}{{\operatorname{PIC}}}

\newcommand{\inc}{{\operatorname{inc}}}
\newcommand{\povm}{{\operatorname{POVM}}}
\newcommand{\Cb}{\mathds C}
\newcommand{\Rb}{\mathds R}
\newcommand{\E}{\textbf{E}}
\renewcommand{\P}{\textbf{P}}
\newcommand{\Ec}{\mathcal E}
\newcommand{\Ect}{{\mathtiny\Ec}}
\newcommand{\Fc}{\mathcal F}
\newcommand{\Ic}{\mathcal{I}}
\newcommand{\Mc}{\mathcal{M}}
\newcommand{\id}{\operatorname{id}}
\newcommand{\rel}{\operatorname{rel}}
\newcommand{\He}{\H_\Ect}
\newcommand{\Pie}{\Pi_\Ect}
\newcommand{\Se}{\Sc_\Ect}

\makeatother 

\begin{document}
\title{Resource theory of coherence based on positive-operator-valued measures}
\author{Felix Bischof}
\email{felix.bischof@hhu.de}
\author{Hermann Kampermann}\author{Dagmar Bru\ss}
\affiliation{Institut f\"ur Theoretische Physik III, Heinrich-Heine-Universit\"at D\"usseldorf,
Universit\"atsstra\ss e 1, D-40225 D\"usseldorf, Germany}
\date{\today}

\begin{abstract}
Quantum coherence is a fundamental feature of quantum mechanics and an underlying requirement for most quantum information tasks. In the resource theory of coherence, incoherent states are diagonal with respect to a fixed orthonormal basis, i.e., they can be seen as arising from a von Neumann measurement.
Here, we introduce and study a generalization to a resource theory of coherence defined with respect to the most general quantum measurement, i.e., to an arbitrary positive-operator-valued measure (POVM).
We establish POVM-based coherence measures and POVM-incoherent operations which coincide for the case of von Neumann measurements with their counterparts in standard coherence theory.
We provide a semidefinite program that allows to characterize interconversion properties of resource states, and exemplify our framework by means of the qubit trine POVM, for which we also show analytical results.
\end{abstract}
\maketitle

Quantum resource theories (QRTs)~\cite{brandao2015reversible,liu2017resource,chitambar2018quantum} 
provide a structured framework in which  quantum properties such as entanglement,
coherence and purity are described in a quantitative way. Every QRT is based on the notions of free states (which do not contain the resource) and  free operations (which cannot create the resource). Building on these basic constituents, QRTs allow to determine, for a given quantum state, its amount of the resource under consideration, the optimal distillation of the resource, and the possibility of interconversion between resource states via free operations.

In recent years, the resource theory of quantum coherence has received much attention~\cite{baumgratz2014quantifying,winter2016operational,streltsov2016quantum,streltsov2018maximal}.
In the standard resource theory of coherence, the free states or incoherent states
are states that are diagonal in a fixed orthonormal basis of a $d$-dimensional Hilbert space $\H$. 
Incoherent states $\rho_{\text{inc}}$ can thus also be seen as arising from a von Neumann measurement $\P=\{P_i\}$ in this basis, i.e., $\rho_{\inc} = \sum_i^dP_i \sigma P_i$ for some state $\sigma\in\Sc$, where $\Sc$ denotes the set of quantum states on $\H$, and the measurement operators $P_i$ are mutually orthogonal projectors of rank one and form a complete set, i.e., $\sum_i^d P_i=\1$. Coherent states are those which are not of the above form.
This notion of coherence has been generalized in two directions. In~\cite{theurer2017resource}, the requirement of orthogonality of the basis was lifted. In~\cite{aberg2006quantifying},  \r{A}berg proposed a framework that can be seen as the definition of coherence with respect to a general projective measurement, where the mutually orthogonal measurement operators $P_i$ may be of higher rank. In this generalized resource theory of coherence the free states are block-diagonal.

It is an important question whether the notion of coherence as an intrinsic quantum property of states can be further extended and formulated with respect to the most general quantum measurements, i.e., positive-operator-valued measures (POVMs). In this letter, we answer this question in the affirmative by introducing a resource theory of quantum state coherence based on arbitrary POVMs. More precisely, we establish a \emph{family} of POVM-based resource theories of coherence, as each POVM leads to a different resource theory. In the special case of rank-1 orthogonal projective measurements, our theory coincides with standard coherence theory. A motivation for our work is the fact that POVMs are generally advantageous compared to projective measurements, see \cite{oszmaniec2017simulating} for a survey. In addition, our approach will identify the resource that is necessary to implement experimentally a general measurement on a given state. 

For a POVM-based coherence theory, the first challenge is to identify a meaningful notion of free, POVM-incoherent, states. This is achieved via the Naimark theorem~\cite{peres2006quantum,decker2005implementation} which states that any POVM can be extended to a projective measurement in a larger space. Our concept of POVM-coherence of states in $\Sc$ is linked to a generalized resource theory of coherence from \cite{aberg2006quantifying} in the extended (Naimark) space, for which we denote the set of states as $\Sc'$. We will show that our resource theories are well-defined as they do not depend on the choice of Naimark extension. 

Conceptually, our work describes a novel way to construct resource theories. Quantum states and operations from the system space are embedded into a larger space which is equipped with a resource theory, providing a derivated resource theory on the original space. For this reason, our work does not follow the standard construction method for a resource theory: our starting point is the definition of a POVM-based coherence measure, from which we construct free states and operations (rather than following the usual reverse order). We then provide a semidefinite program that characterizes all POVM-incoherent operations, making them accessible for efficient numerical computation. Finally, we apply our framework to the example of the qubit trine POVM, for which we study the POVM-based coherence measure and characterize all incoherent unitaries. 
 
In the following, we present our main results and their interpretation. Technical details and proofs from every section of the main text can be found in the corresponding section of the Supplemental Material~\cite{supp}.

\paragraph*{POVM and Naimark extension---}
A POVM on $\H$ with $n$ outcomes is a set $\E=\{E_i\}_{i=1}^{n}$ of positive 
semidefinite operators $E_i\geq0$, called POVM elements or effects, which satisfy $\sum_i^nE_i=\1$.
The probability to obtain the $i$-th outcome when measuring $\rho$ is given by $p_i(\rho)=\tr[E_i\rho]$. We denote by $\{A_i\}$ a set of measurement operators of $\E$, i.e., $E_i=A_i^\dagger A_i$. Each measurement operator  $A_i$ is only fixed up to a unitary $U_i$, as the transformation $A_i \to U_i A_i$ leaves $E_i$ invariant. 
The $i$-th post-measurement state for a given $A_i$ is $\rho_i=\frac1{p_i}A_i\rho A_i^\dagger$.

Let us remind the reader that according to the Naimark theorem~\cite{peres2006quantum,decker2005implementation}, every POVM $\E=\{E_i\}_{i=1}^{n}$ on $\H$, if embedded in a larger Hilbert space, the Naimark space $\H'$ of dimension $d'\geq d$, can be extended to a projective measurement $\P=\{P_i\}_{i=1}^{n}$ on $\H'$. The most general way to embed the original Hilbert space $\H$ into $\H'$ is via a direct sum, requiring
\begin{align}\label{eq:naimark}
\tr[E_i\rho]=\tr[P_i(\rho\oplus 0)]
\end{align}
to hold for all states $\rho$ in $\Sc$, where $\op$ denotes the orthogonal direct sum, and $0$ is the zero matrix of dimension $d'-d$. We call any projective measurement $\P$ which fulfills Eq.~(\ref{eq:naimark}) a Naimark extension of $\E$.

The embedding into a larger-dimensional space can also be performed via the so-called \emph{canonical} Naimark extension~\cite{wilde2013sequential,sparaciari2013canonical}: 
one attaches an ancilla or probe (initially in a fixed state $\proj{1}$) via a tensor product. 
We denote the map that performs the embedding by $\Ec[\rho]=\rho\ot\proj{1}$ and the space of embedded states by $\Se=\Ec[\Sc]$.
A suitable global unitary $V$ describes the interaction between system and probe such that the resulting state is $\rho'\defeq V(\rho\ot\proj{1})V^\dagger$. 
A von Neumann measurement on the probe leads to the same probabilities $p_i$ as the POVM if
\begin{align}\label{eq:canonaim}
\tr[E_i\rho]&= \tr[(\1\ot\proj{i})\rho']=\tr[P_i(\rho\ot\proj{1})]
\end{align}
holds for all states $\rho$ in $\Sc$. Here we have included the unitary $V$ into the projective measurement, i.e., $P_i:=V^\dagger(\1\ot\proj{i})V$. Thus, $P_i$ has rank $d$.
This type of Naimark extension is not optimal in terms of smallest additionally
required dimension~\cite{chen2007ancilla}, but its structure allows for a simpler derivation of general results, and directly describes the possibility to implement a POVM in an experiment. Both described types of Naimark extensions are not unique.

\paragraph*{Resource theory of block coherence---}
\r{A}berg~\cite{aberg2006quantifying} introduced general measures for the degree of superposition in a mixed quantum state with respect to orthogonal decompositions of the underlying Hilbert space, thus pioneering the resource theory of coherence. Here we translate his work into the present-day language of resource theories and will refer to it as resource theory of block coherence.

The set $\Ic$ of block-incoherent (or free) states $\rho_{\inc}$ arises via a projective measurement $\P=\{P_i\}_{i=1}^{n}$ on the set of quantum states $\Sc$, 
namely~\cite{aberg2006quantifying}
\begin{align}
\rho_{\inc}=\sum_iP_i\sigma P_i = \Delta[\sigma],\quad \sigma\in\Sc,
\end{align}
where the rank of the orthogonal projectors $P_i$ is arbitrary, and we have defined the \emph{block-dephasing map} $\Delta$. In this framework, coherence is not ``visible'' within a subspace given by the range of $P_i$, but only across different subspaces.
If all $P_i$ have rank 1, the standard resource theory of coherence is recovered.
Note that here we have intentionally chosen the same symbol $P_i$ as in Eq.~(\ref{eq:canonaim}), as we shortly identify the two.

We refer to the largest class of (free) operations that cannot create block coherence as (maximally) \emph{block-incoherent} (BIC) operations. A map $\Lambda_{\bic}$ on $\Sc$ is element of this class iff it maps any block-incoherent state to a block-incoherent state, i.e.,
\begin{align}\label{mio}
\Lambda_{\bic}[\Ic]\subseteq \Ic,
\end{align}
or equivalently $\Lambda_{\bic}\circ\Delta=\Delta\circ\Lambda_{\bic}\circ\Delta$. In standard coherence theory this class is referred to as maximally-incoherent operations $(\mio)$.

The amount of block coherence contained in a state $\rho$ with respect to a projective measurement {\P} can be quantified by suitable measures. We call a realvalued positive function $C(\rho,\P)\geq0$ a block-coherence measure iff it fulfills
\begin{enumerate}[label=\roman{*}.]
\item \emph{Faithfulness:} $C(\rho,\P)=0 \Leftrightarrow \rho\in\Ic$,
\item \emph{Monotonicity:}
$C(\Lambda_{\bic}[\rho],\P)\leq C(\rho,\P)$ for all $\Lambda_{\bic}$.
\end{enumerate}
Several block-coherence measures were introduced in~\cite{aberg2006quantifying}, and a general class of measures can be derived from distances that are contractive under quantum operations~\cite{supp}. An important example for a suitable measure is the relative entropy of block coherence, defined as 
\begin{align}\label{relentblkcoherence}
C_{\rel}(\rho,\P)=\min_{\sigma\in\Ic}S(\rho||\sigma)=S(\Delta[\rho])-S(\rho),
\end{align}
where $S(\rho||\sigma)=\tr[\rho\log\rho-\rho\log\sigma]$ denotes the quantum relative entropy and $S(\rho)=-S(\rho||\1)$ is the von Neumann entropy. In standard coherence theory, the relative entropy of coherence has several important operational meanings~\cite{winter2016operational,yuan2015intrinsic,yuan2016interplay}, e.g., it quantifies the distillable coherence and coherence cost under the class $\mio$~\cite{streltsov2016quantum}.

\paragraph*{POVM-based coherence measures---}
The main idea of our approach is to define the coherence of a state $\rho$ with respect to the POVM $\E$ via its canonical Naimark extension. This concept is visualised in Fig.~\ref{fig:povmcoh}.

\begin{figure}[!h]
\centering
\includegraphics[width=\columnwidth]{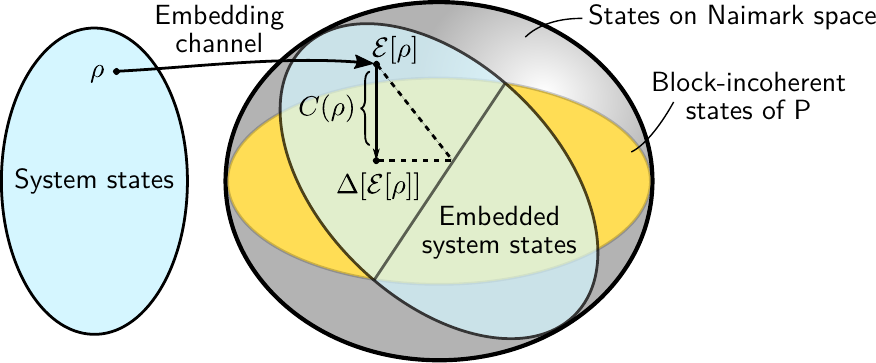}
\caption{\label{fig:povmcoh}
We introduce a resource theory of POVM-based coherence by making use of the Naimark construction. Quantum states $\rho$ are embedded as $\Ec[\rho]= \rho\ot\proj{1}$ to act on a higher-dimensional Hilbert space (Naimark space). The POVM $\E$ is extended to a projective measurement $\P$ on the Naimark space, which defines a set of block-incoherent states $\Ic$. 
The POVM-coherence measure $C(\rho,\E)$ is the distance between $\Ec[\rho]$ and its projection $\Delta[\Ec[\rho]]$ onto block-incoherent states.}
\end{figure}

\begin{definition}[POVM-based coherence measure]\label{def:measure}
Let $C(\rho',\P)$ be a unitarily-invariant block-coherence measure on $\Sc'$.
The POVM-based coherence measure $C(\rho,\E)$ for a state $\rho$ in $\Sc$
is defined as the block coherence of the embedded state $\Ec[\rho]=\rho\ot\proj{1}$ 
with respect to a canonical Naimark extension $\P$ of the POVM $\E$, namely
\begin{align}\label{cohmeasure}
C(\rho,\E)\defeq C(\Ec[\rho],\P),
\end{align}
where the constraint in Eq.~(\ref{eq:canonaim}) has to hold. ---It is straightforward to generalize this definition also to the most general Naimark extension from Eq.~(\ref{eq:naimark}).
\end{definition}
Here, unitarily-invariant means that $C(\rho',\P)=C(U\rho' U\ad,U\P U\ad)$ holds for all unitaries $U$ on $\H'$. This property ensures that $C(\rho,\E)$ is well-defined, as it is consequently invariant under a change of measurement operators $A_i\to U_iA_i$, with unitary $U_i$~\cite{supp}.

In this letter, we focus on the relative entropy measure. It has the crucial property that in this case the POVM-based coherence measure is independent of the choice of Naimark extension (regarding the dimension and the form) for its definition~\cite{supp}.

\begin{lemma}[Analytical form of a POVM-based coherence measure]\label{lem:pbmeas}
The relative entropy of POVM-based coherence $C_{\rel}(\rho,\E)$ does not depend on the choice of Naimark extension for its definition and admits the following form,
\begin{align}
C_{\rel}(\rho,\E)=H(\{p_i(\rho)\})+\sum_ip_i(\rho)S(\rho_i)-S(\rho),
\end{align}
with $p_i(\rho)=\tr[E_i\rho]$, $\rho_i=\frac1{p_i}A_i\rho A_i^\dagger$, and the Shannon entropy $H(\{p_i(\rho)\})=-\sum_ip_i\log p_i$. In the special case of $\E$ being a von Neumann measurement, i.e., $E_i=\proj{i}$, $C_{\rel}(\rho,\E)$ equals the standard relative entropy of coherence.
\end{lemma}
The independence property holds because the eigenvalues of $\Delta[\Ec[\rho]]$  are the same for any two Naimark extensions used to define $\Delta$ and because the von Neumann entropy is a function solely of the eigenvalues of its argument~\cite{supp}.

\paragraph*{Minimal and maximal POVM-based coherence---}
We show in~\cite{supp} that for an $n$-outcome POVM $\E$ the bounds $0\leq C_{\rel}(\rho,\E) \leq \log n$ hold, see also~\cite{chitambar2016assisted}.
However, there exist POVMs for which one or both of these bounds cannot be attained for any quantum state. First, let us discuss the upper bound, i.e., maximal coherence: the convexity of $C_{\rel}$ implies that its maxima are attained by the pure states that lead to the highest entropy of measurement outcomes. 

Now, we address the lower bound. We can can characterize POVM-incoherent states (i.e., states with zero POVM coherence) via the POVM elements as follows.

\begin{lemma}[Characterization of POVM-incoherent states]\label{lem:sysincoherent}
A state $\rho$ is POVM-incoherent with respect to a POVM
$\E=\{E_i\}_{i=1}^{n}$  iff
 the following holds: 
\begin{align}
E_i \rho E_j=0 \quad \forall\ i\neq j\in\{1,\dotsc,n\} \ .
\end{align}
\end{lemma}
This condition is readily obtained from the definition of POVM-incoherence for the canonical Naimark extension. It generalizes the requirement of vanishing off-diagonal elements for incoherent states in standard coherence theory. In contrast to standard coherence theory states of the form $\rho=\sum_i \sqrt{E_i}\sigma \sqrt{E_i}$ are not necessarily incoherent, due to the possible non-orthogonality of the measurement operators $\sqrt{E_i}$. 
Indeed, for many interesting POVMs the set of incoherent states $\Ic_\povm$ is 
empty, since $C_{\rel}(\rho,\E)>0$, e.g., for rank-1 POVMs with no effect that is a projector~\cite{massar2007uncertainty} like the trine POVM which we discuss in detail below. The set $\Ic_\povm$ may be empty, but the states with {\em minimal} POVM-based coherence nevertheless form a convex set $\Mc$, due to the convexity of $C_{\rel}(\rho,\E)$ in $\rho$~\cite{aberg2006quantifying,vedral1998entanglement}. Interestingly, the maximally mixed state $\rho=\frac\1d$ is not necessarily contained in $\Mc$ for any POVM.

\paragraph*{POVM-incoherent operations---}
The final main ingredient of our resource theory are quantum operations that cannot increase POVM-based coherence, i.e., free operations. They exist for any POVM, ensuring a nontrivial resource theory. The second requirement in the following definition ensures that the operation is trace-preserving.
Throughout this paper, maps acting on the larger space $\Sc'$ are denoted as $\Lambda'$, while maps acting on the original system $\Sc$ will be called $\Lambda$.

\begin{definition}[POVM-incoherent (free) operations]\label{def:povmmio}
Let $\E$ be a POVM and $\P$ any Naimark extension of it. Let $\Lambda'$ 
be a completely positive trace-preserving map on $\Sc'$ that is
\begin{enumerate}[label=\roman{*}.]
\item \emph{Block-incoherent:}
$\Lambda'$ is maximally block-incoherent (BIC) with respect to $\P$,
see Eq.~(\ref{mio}).
\item \emph{Subspace-preserving:}
$\Lambda'[\Se]\subseteq \Se$ for the subset $\Se\subseteq \Sc'$ of embedded system states.
\end{enumerate}
We call the restricted channel $\Lambda'|_{\Se}$ an embedded POVM-incoherent operation, and $\Lambda_\pic\defeq\Ec^{-1}\circ\Lambda'\circ\Ec$ a POVM-incoherent operation.
\end{definition}

This definition of POVM-incoherent operations implies that they cannot increase the POVM-based coherence of any state.

\begin{lemma}[Operations as defined in Def.~\ref{def:povmmio}
can indeed not increase POVM-based coherence]\label{lem:povmmon}
Let $\Lambda_\pic$ be a POVM-incoherent operation of the POVM $\E$. Then, for every POVM-based coherence measure $C(\rho,\E)$ it holds that 
\begin{align}
C(\Lambda_\pic[\rho],\E)\leq C(\rho,\E).
\end{align}
\end{lemma}

For any measurement, we can characterize the set of POVM-incoherent operations by a semidefinite program (SDP), since these operations are defined solely by linear 
conditions (\romannumeral1, \romannumeral2\ and trace-preservation) and semidefinite conditions (complete positivity).

\begin{theorem}[Characterization of POVM-incoherent operations]\label{thm:povmmio}
The set of POVM-incoherent operations is independent of the chosen Naimark extension and can be characterized by a semidefinite feasibility problem (SDP). In the 
case of von Neumann measurements, POVM-incoherent operations are equivalent 
to MIO maps of the standard coherence theory. 
\end{theorem}
The independence property holds because for POVM-incoherent channels $\Lambda_\pic$ only the action of $\Lambda'$ on $\Se$ is relevant, which can be shown to be equal for any two Naimark extension, as these are connected by an isometry when acting on embedded system states~\cite{supp}. 

Regarding the interconversion of resource states in our POVM-based coherence theory,
we can employ the SDP characterization of POVM-incoherent operations $\Lambda_\pic$ for a POVM $\E$ to determine numerically the maximally achievable fidelity $F_{\max}(\rho,\sigma)=\max_{\Lambda_\pic}F(\Lambda_\pic[\rho],\sigma)$ between a target state $\sigma$ and $\Lambda_\pic[\rho]$, see the Supplemental Material~\cite{supp}.

\paragraph*{Example: qubit trine POVM---}

As an example, we analyze the case of the qubit trine POVM $\E^{\operatorname{trine}}=\{\frac23\proj{\phi_i}\}_{i=1}^3$, with measurement directions $\ket{\phi_i}=1/\sqrt{2}(\ket{0}+\omega^{i-1}\ket{1})$, where $\omega=\operatorname{exp}(2\pi i/3)$.
The corresponding POVM-based coherence of \emph{pure} states is illustrated in Fig.~\ref{fig:trineplot} (left).
For the qubit trine POVM  there are two states with maximal POVM-coherence $C_{\rel}^{\max}=\log3$, namely $\ket{\Psi_{\operatorname m}}\in\{\ket{0},\ket{1}\}$.
The set $\Mc$ of states with minimal POVM-based coherence $C_{\rel}^{\min}=\log3-1$
contains solely the maximally mixed state, i.e., $\Mc=\{\frac{\1}2\}$. 

\begin{figure}[!h]
\centering
\includegraphics[width=0.5\columnwidth,valign=B]{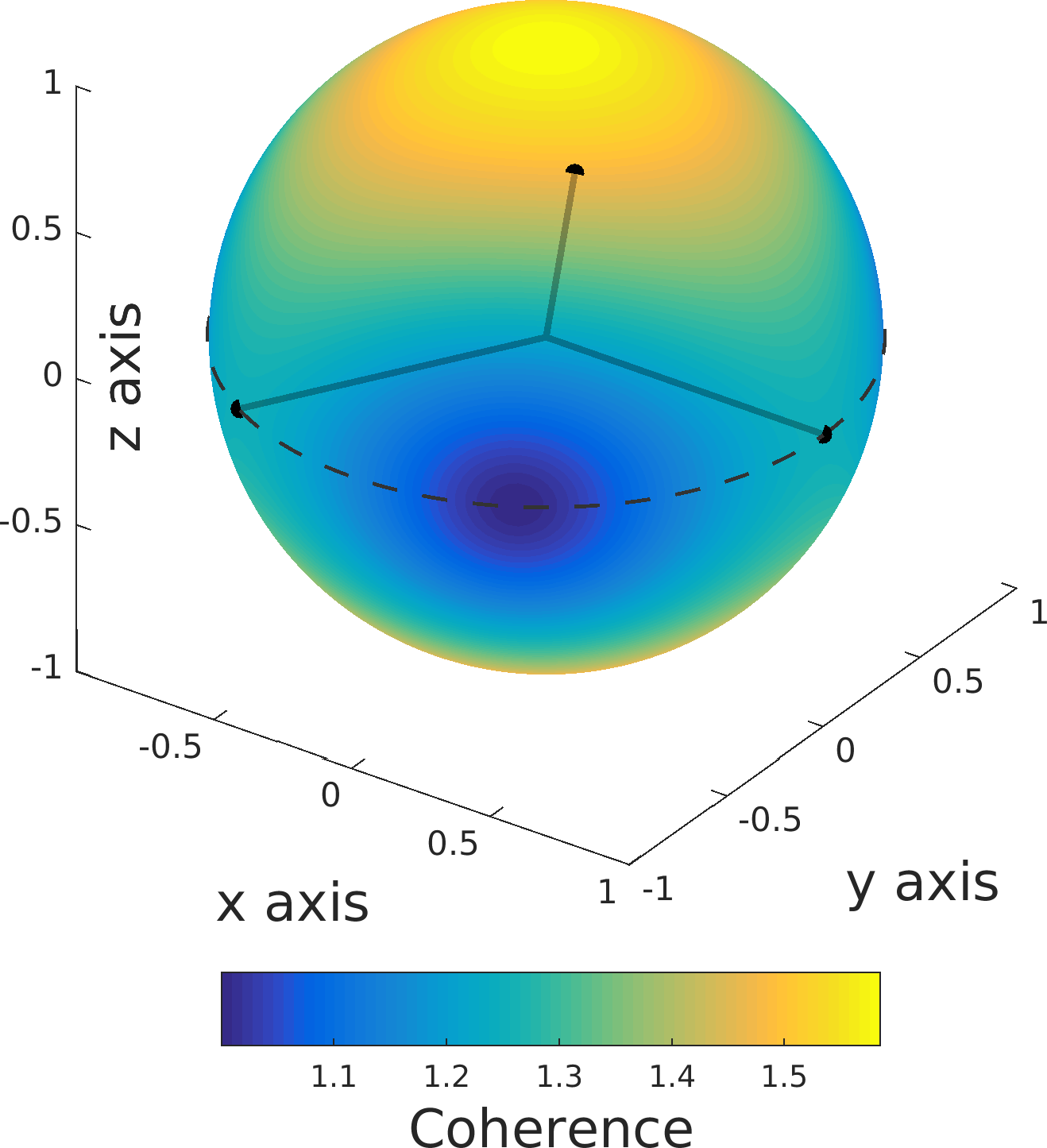}
\includegraphics[width=0.5\columnwidth,valign=B]{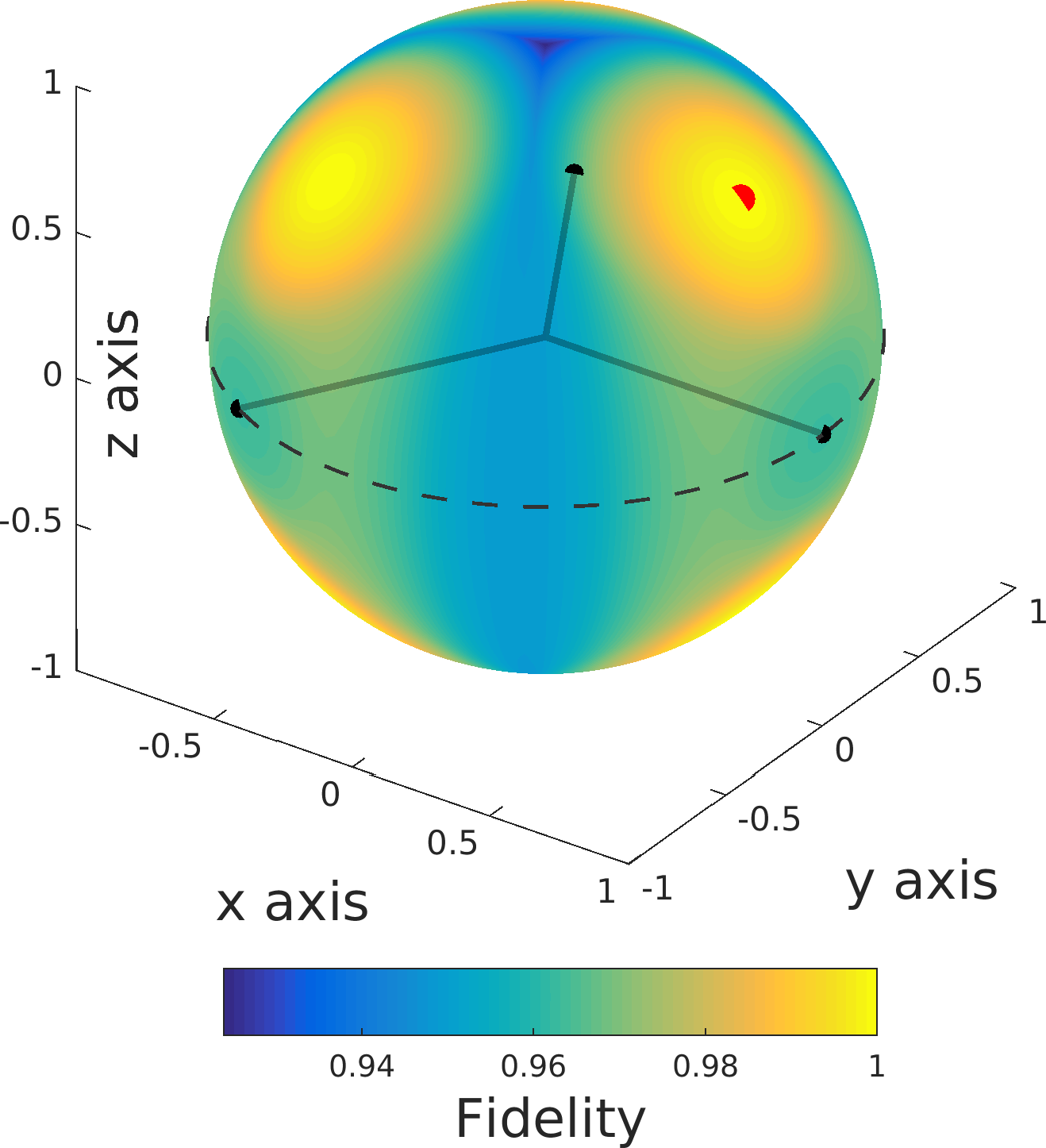}
\caption{\label{fig:trineplot}
POVM-based coherence theory for qubit states with respect to the trine POVM 
$\E^{\operatorname{trine}}$, in the Bloch sphere representation. Gray lines indicate the three measurement directions (see main text for details). 
\emph{Left:}  POVM-based coherence of pure qubits (surface of sphere). 
The states $\ket{0}$ and $\ket{1}$ have maximal POVM-based coherence $C=\log3$. The Bloch vectors of the three  states with the lowest pure-state coherence $C=1$ are antipodal to the measurement directions. 
\emph{Right:} Maximally achievable conversion fidelity 
$F_{\max}(\rho,\sigma)=\max_{\Lambda_\pic}F(\Lambda_\pic[\rho],\sigma)$ between 
an initial state $\rho$ (red dot) subjected to POVM-incoherent operations $\Lambda_\pic$ and a target state $\sigma$  on the sphere surface. Here, $\rho=\proj{\psi}$ with $\ket{\psi}=\cos(\frac\pi8)\ket{0}+\sin(\frac\pi8)\ket{1}$. Only states in the orbit of $\ket{\psi}$ under the six POVM-incoherent unitaries can be reached with unit fidelity, as depicted by the yellow spots.}
\end{figure}

Regarding POVM-incoherent (free) operations, the free \emph{unitary} operations can be fully characterized: there exist exactly six POVM-incoherent unitaries  $U_i^{\operatorname{trine}}$. They correspond to the rotations on the Bloch sphere that map the trine star to itself, i.e., the symmetry group of the equilateral triangle. 
In standard coherence theory the measurement map $\rho\to\Delta[\rho]$ is incoherent. However, for a general POVM the measurement map $\rho\to\sum_i\sqrt{E_i}\rho\sqrt{E_i}$ is not necessarily POVM-incoherent with respect to $\E$ as one can find POVMs for which the map increases the coherence of a state. Notably, for the trine POVM $\E^{\operatorname{trine}}$ the SDP verifies that the measurement map is indeed POVM-incoherent.
As to  conversion properties, every qubit state $\rho$ can be obtained deterministically by applying some POVM-incoherent operation to a maximally coherent state 
$\ket{\Psi_{\operatorname m}}\in\{\ket{0},\ket{1}\}$.
By applying the SDP, we have numerical evidence that given a state $\ket{\psi}\neq\ket{\Psi_{\operatorname m}}$, the only pure states that can be obtained from it with certainty via POVM-incoherent operations are in the orbit $\{U_i^{\operatorname{trine}}\ket{\psi}\}$ under the six trine-incoherent unitaries. An example for the conversion fidelity when starting from an initial state with less than maximal resource is shown in Fig.~\ref{fig:trineplot} (right).

\paragraph*{Conclusion and Outlook---}
We have introduced a familiy of resource theories which quantify the coherence of a
quantum state with respect to a given arbitrary POVM. To define these resource theories, we have embedded the states and operations into a higher-dimensional Naimark space. There, the POVM can be extended to a projective measurement for which a resource theory of block coherence exists~\cite{aberg2006quantifying}.
The restriction to the embedded original space led to the characterization of free states, free operations and resulting conversion properties within the POVM-based resource theories. As a POVM-based coherence measure we have studied the relative-entropy-based measure which is invariant under the choice of Naimark extension. For the case of von Neumann measurements, POVM-based coherence measures and POVM-incoherent operations reduce to their counterparts in standard coherence theory. 

The canonical Naimark extension is realized by coupling the state to a probe, performing a global unitary and measuring the probe. One can view the probe as a measurement apparatus: then, the POVM-based coherence measure quantifies the global-state coherence generated by the unitary on the state with respect to a fixed basis of the apparatus. Thus, our measure can be interpreted as the resource that is necessary to implement the POVM on a state via the Naimark extension.
Note that in general a part of this coherence is used to generate classical randomness for the mixing of the POVM. If the experimenter is able to perform statistical mixtures of measurements, certain POVMs can be implemented with less resource.
Also note other works that elucidate the role of quantum resources in the Naimark extension. In~\cite{streltsov2011linking} it was shown that if a composite system carries discord, any local von Neumann measurement necessarily creates entanglement between the measurement apparatus and the system. Ref.~\cite{jozsa2003entanglement} discusses the minimal amount of average entanglement contained in the canonical Naimark effects $P_i$ of a rank-1 POVM.

Several open questions should be addressed in the future. First, it is not clear whether a characterization of POVM-incoherent operations without reference to the Naimark space is possible. A necessary condition is given by $\Lambda_\pic[\Mc]\subseteq \Mc$, where $\Mc$ is the set of states with minimal POVM-based coherence. For projective measurements, this property is also sufficient, while the qubit trine POVM is a counterexample for this property to be sufficient in general: there, the state with minimal coherence is proportional to the identity, which is left invariant by all unital maps. However, almost all unitary channels can increase the POVM-based coherence~\cite{supp}. 
A further interesting problem is whether operational interpretations of regular coherence measures \cite{yuan2015intrinsic,biswas2017interferometric,napoli2016robustness} can be generalized to POVM-coherence, by e.g., generalizing the considered quantum information protocols. Finally, several generalizations of our framework are possible. One can study POVM-coherence measures beyond the one based on the relative entropy. Also, in analogy to regular coherence theory, one can introduce the POVM-coherence equivalents of the channel classes IO, SIO etc. \cite{streltsov2016quantum} which are subsets of the POVM-incoherent operations, and study the corresponding conversion properties.

We acknowledge financial support from the German Federal Ministry of Education and Research (BMBF).
FB gratefully acknowledges support from Evangelisches Studienwerk Villigst and from Strategischer Forschungsfonds (SFF) of the Heinrich Heine University D\"usseldorf. 

\bibliography{CohBib}{}
\bibliographystyle{apsrev}

\newpage
\part*{Supplemental Material}\label{supp}

In the following we provide the technical details and proofs that complement the main text. Moreover, we provide examples and plots to illustrate the general results.
In part \ref{app:naimark} we give a detailed description of the Naimark extension of a POVM. Subsequently, in part~\ref{app:blkcoh} we describe the resource theory of block coherence~\cite{aberg2006quantifying} in a way that is analogous to standard coherence theory.
The rest of the Supplemental Material is devoted to our resource theory of POVM-based coherence, which we formulate in the main text. In part~\ref{app:pbcohmeas}, we discuss general POVM-based coherence measures, including the relative entropy of POVM-based coherence on which we focus for the remainder of the paper. 
Achievable lower and upper bounds for the POVM-based coherence of quantum states are discussed in part~\ref{app:minmaxcoh}.
In part~\ref{app:povmicops} we prove properties of POVM-incoherent operations, i.e., free operations, and present explicitly the semidefinite program that characterizes them. Moreover, we study the conversion of resource states under free operations.
Finally, in part~\ref{app:qubitcoh} we exemplify all general results by means of the qubit trine POVM and provide analytical results of its POVM-based coherence theory.

\begin{table}[h]
\begin{centering}
\begin{tabular}{|c|l|}\hline
\textbf{Symbol} & \textbf{Explanation} \tabularnewline\hline
$\H$ & $d$-dimensional Hilbert space \tabularnewline\hline
$\H'$ & $d'$-dimensional (Naimark) Hilbert space \tabularnewline\hline
$\Sc$ & set of (system) quantum states on $\H$ \tabularnewline\hline
$\Sc'$ & set of quantum states on $\H'$ \tabularnewline\hline
$\Ec$ & embedding channel, from $\Sc$ to $\Sc'$ \tabularnewline\hline
$\He$ & subspace of $\H'$ of embedded state vectors \tabularnewline\hline
$\Pie$ & orthogonal projector onto $\He$ \tabularnewline\hline
$\Se$ & subset of $\Sc'$ of embedded system states \tabularnewline\hline
$\Omega$ & projector onto operators on $\He$ \tabularnewline\hline
$\E$ & POVM on $\H$ \tabularnewline\hline
$\{A_i\}$ & set of measurement operators of $\E$ \tabularnewline\hline
$V$ & Naimark interaction unitary on $\H'$ \tabularnewline\hline
$\P$ & Naimark extension of $\E$ on $\H'$ \tabularnewline\hline
$\Ic$ & set of block-incoherent states of $\P$ \tabularnewline\hline
$\Lambda'_\bic$ & block-incoherent operation on $\Sc'$ \tabularnewline\hline
$\Lambda'|_{\Se}$ & embedded POVM-incoherent operation \tabularnewline\hline
$\Lambda_\pic$ & POVM-incoherent operation on $\Sc$ \tabularnewline\hline
\end{tabular}
\par\end{centering}
\caption{Notation used throughout this work.}
\end{table}

\appendix

\section{POVM and Naimark extension}\label{app:naimark}

In this part, we provide the details and construction of the Naimark extension of a POVM $\E$. Under a general Naimark extension we understand any projective measurement $\P=\{P_i\}$ on $\H'$, which fulfills Eq.~(\ref{eq:naimark}). Consequently, the uppper left $d\times d$ block of the Naimark extension effect $P_i$ coincides with the POVM effect $E_i$, i.e., for $\Pie=\1_d\op0_{d'-d}$
it holds that 
\begin{align}\label{aux4}
E_i\op0_{d'-d}=\Pie P_i\Pie.
\end{align}
Therefore, it is convenient to embed system operators $X$ on $\H$ into the Naimark space $\H'$ via the embedding map $\Ec[X]\defeq X\op0$ and call $\He\defeq\{\ket{\psi}\op0:\ket{\psi}\in\H\}$ the embedded state space. Now, we only need to consider (embedded) operators on the Naimark space $\H'$. The construction of a projective measurement fulfilling Eq.~(\ref{aux4}) is straigthforward. We provide details for the \emph{canonical} Naimark extension, for which the Naimark space has product form $\H'=\H\ot\H_R$, with $\H_R$ being the probe's state space. In this case, we employ the embedding map $\Ec[X]\defeq X\ot\proj{1}$ and $\He=\H\ot\ket{1}$. The canonical Naimark extension is generally not of the smallest possible dimension $d'_{\min}=\sum_i\rank E_i$~\cite{chen2007ancilla}. In section~\ref{app:pbcohmeas} (\ref{app:povmicops}) we show that POVM-based coherence measures (POVM-incoherent operations) are independent of the choice of Naimark extension.

Let $\E=\{E_i\}_{i=1}^{n}$ be an $n$-outcome POVM on $\H$, and $\{A_i\}$ a set of measurement operators for $\E$. Any measurement operator can be written as 
$A_i=U_i\sqrt{E_i}$ for some unitary operator $U_i$. Let $\{\ket{i}\}$ be an orthonormal basis of the probe space $\H_R$ and define the operator 
\begin{align}
\tilde V = \sum_{i=1}^{n}A_i\ot\ketbra{i}{0},
\end{align}
which is an isometry from $\He$ to $\H'$, i.e., it fulfills $\tilde V\ad\tilde V=\1_d$ on $\He$ as a consequence of the normalization of the POVM.
The isometry $\tilde V$ can be extended to a unitary $V$ on $\H'$ by completing the set of orthonormal column vectors (lying in $\im\tilde V\subseteq\H'$) to an orthonormal basis, i.e., by filling up the columns of the $nd\times d$ matrix to an $nd\times nd$ matrix with orthonormal column vectors.
Now, we can parameterize the unitary by operators $A_{i,a}$ on $\H$ as
\begin{align}
V = \sum_{i,a=1}^{n}A_{i,a}\ot\ketbra{i}{a},
\end{align}
where $A_{i,1}=A_i$. To ensure the unitary condition $V\ad V=VV\ad=\1_{d'}$, these operators need to fulfill
\begin{align}
&\sum_{i}A_{i,a}\ad A_{i,b}=\delta_{a,b}\1_d,\quad\trm{and}\ncl
&\sum_{a}A_{i,a} A_{j,a}\ad=\delta_{i,j}\1_d.
\end{align}
Finally, the canonical Naimark extension $\P=\{P_i\}_{i=1}^{n}$ of $\E$ is defined as
\begin{align}
P_i &\defeq V\ad\1\ot\proj{i}V \ncl
&=\sum_{a,b} A_{i,a}\ad A_{i,b}\ot\ketbra{a}{b},
\end{align}
which is a rank-$d$ projective measurement, i.e., $\rank P_i=d$ and $P_iP_j=\delta_{i,j}P_j$. Projecting this measurement onto the embedded state space $\He$ with the projector $\Pie=\1\ot\proj{1}$ yields the embedding of the POVM $\E$, 
\begin{align}
\Pie P_i\Pie=A_{i,1}\ad A_{i,1}\ot\proj{1}=E_i\ot\proj{1}=\Ec[E_i].
\end{align}
This property implies that 
\begin{align}
\tr[E_i\rho]=\tr[P_i(\rho\ot\proj{1})],
\end{align}
and thus, up to an exchange of the two subsystems the Naimark extension property from Eq.~(\ref{eq:naimark}).

\section{Resource theory of block coherence}\label{app:blkcoh}

In this part, we supplement details of the resource theory of block coherence as described in the main text. This theory was introduced in~\cite{aberg2006quantifying} and we formulate it in a modern way that is analogous to the resource theory of coherence~\cite{streltsov2016quantum}. 

Let $\P=\{P_i\}_{i=1}^{n}$ be a projective measurement on $\H$. Resource-free states, called block-incoherent, are block-diagonal with respect to $\P$ and belong to the set 
\begin{align}
\Ic=\{\rho_{\inc}=\sum_iP_i\sigma P_i:\sigma\in\Sc\}.
\end{align}
Therefore, block-incoherent states are characterized as the image of the block-dephasing operator $\Delta[\sigma]=\sum_iP_i\sigma P_i$ applied to any state $\sigma$.
By the mutual orthogonality of the effects $P_i$, block-incoherent states are also characterized by the condition
\begin{align}
\rho\in\Ic \Leftrightarrow P_i\rho P_j=0 \ \ \forall i\neq j\in\{1,\dotsc,n\}.
\end{align}
Let $C(\rho,\P)$ be a block-coherence measure as defined in the main text. Any block-coherence measures obeys a further desirable property, which ensures that it does not depend on the choice of basis within the blocks (subspaces) $\pi_i\defeq\im P_i$.

\begin{proposition}\label{prop:buinv}
Every block-coherence measure $C(\rho,\P)$ as defined in the main text also fulfills
\begin{enumerate}[label=\roman{*}.]  \setcounter{enumi}{2}
\item \emph{Block-unitary invariance:}
$C(U\rho U\ad,\P)= C(\rho,\P)$ with $U=\op_iU_i$, and where $U_i$ is unitary on $\pi_i$.
\end{enumerate}
\end{proposition}
\begin{proof}
The assertion holds since $U$ is a reversible block-incoherent operation and thus the monotonicity property holds with equality. More precisely,
it holds that $P_iU=UP_i$ and therefore $U\Delta[\sigma]U\ad=\Delta[U\sigma U\ad]$. 
Because block-incoherent states have the form $\rho=\Delta[\sigma]$, we conclude that $U\rho U\ad\in \Ic$ for all states $\rho\in\Ic$. Hence, the unitary channel $\rho\to U\rho U\ad$ is a maximally-block-incoherent (BIC) operation.
Since unitary channels are invertible, we can apply the monotonicity property in both directions to obtain equality, from which the assertion follows.
\end{proof}

Several block-coherence quantifier were introduced in \cite{aberg2006quantifying}, and the monotonicity condition was proven for some of them. We consider a general class of a block-coherence quantifier that is obtained from a distance $D(\rho,\sigma)$ via
\begin{align}
C(\rho,\P) \defeq \inf_{\sigma\in\Ic}D(\rho,\sigma).
\end{align} 
Certain properties of the distance measure lead to the block-coherence measure properties.

\begin{proposition}\label{prop:dbcoh}
The distance-based block-coherence quantifier $C(\rho,\P) = \inf_{\sigma\in\Ic}D(\rho,\sigma)$ fulfills
\begin{enumerate} [label=\roman{*}.]  
\item \emph{Positivity and Faithfulness}, if the distance $D(\rho,\sigma)$ is nonnegative and vanishes if and only if $\rho=\sigma$.
\item \emph{Monotonicity}, if the distance is contractive under quantum operations $\Lambda$, that is,
$D(\Lambda[\rho],\Lambda[\sigma])\leq D(\rho,\sigma)$.
\end{enumerate}
\end{proposition}

\begin{proof}
The proof of the second assertion is analogous to the case of coherence theory \cite{baumgratz2014quantifying,streltsov2016quantum}.
For convenience, we outline it here,
\begin{align}
C(\rho,\P)&=\inf_{\sigma\in\Ic}D(\rho,\sigma) = D(\rho,\sigma^*) \ncl
&\geq D(\Lambda_{\bic}[\rho],\Lambda_{\bic}[\sigma^*]) \ncl
&\geq \inf_{\tau\in\Ic} D(\Lambda_{\bic}[\rho],\tau) \ncl
&=C(\Lambda_{\bic}[\rho]),
\end{align}
where $\sigma^*$ denotes a state that achieves the minimum. The first inequality follows from the contractive property of the distance, and the second equality holds because $\Lambda_{\bic}[\sigma^*]\in I$.
\end{proof}

In the following, we focus on the relative-entropy-based block-coherence measure which was introduced in~\cite{aberg2006quantifying}.
The relative entropy of block coherence is defined as
\begin{align}
C_{\rel}(\rho)\defeq\min_{\sigma\in\Ic}S(\rho||\sigma),
\end{align}
where $S(\rho||\sigma)$ denotes the quantum relative entropy. The coherence quantifier $C_{\rel}(\rho)$ is convex and satisfies nonnegativity. Moreover, the monotonicity property is proven in the following Proposition.

\begin{proposition}
The relative entropy of block coherence is a block-coherence measure and admits the following simple form
\begin{align}
C_{\rel}(\rho) = S(\Delta[\rho])-S(\rho).
\end{align}
\end{proposition}
\begin{proof}
The relative entropy is contractive under quantum operations \cite{vedral1998entanglement} and thus $C_{\rel}(\rho)$ satisfies the monotonicity condition because of Prop.~\ref{prop:dbcoh}. The simplified form $C_{\rel}(\rho) = S(\Delta[\rho])-S(\rho)$ was first stated in \cite{aberg2006quantifying}, and is proven analogous to coherence theory \cite{matera2016coherent,chitambar2016assisted}. For the convenience of the reader, we outline the proof. Observe that 
\begin{align}
\tr[\Delta[\rho]\log\sigma_\inc]&=\tr[\rho\Delta[\log\sigma_\inc]] \ncl
&=\tr[\rho\log\sigma_\inc],
\end{align}
for any block-incoherent state $\sigma_\inc$, since the operator logarithm of a nonnegative matrix preserves the block-diagonal structure, as it only acts on the eigenvalues. This implies that 
\begin{align}
S(\rho||\sigma_\inc) &=\tr[\rho\log\rho]-\tr[\rho\log\sigma_\inc] \ncl
&=-\tr[\Delta[\rho]\log\Delta[\rho]]+\tr[\rho\log\rho] \ncl
&\quad+\tr[\Delta[\rho]\log\Delta[\rho]]-\tr[\Delta[\rho]\log\sigma_\inc] \ncl
&=S(\Delta[\rho])-S(\rho)+S(\Delta[\rho]||\sigma_\inc).
\end{align}
The third term is nonnegative, $S(\Delta[\rho]||\sigma_\inc)\geq0$, and therefore the minimum over block-incoherent states is achieved when it vanishes, i.e., $\sigma_\inc=\Delta[\rho]$.
\end{proof}

\section{POVM-based coherence measures}\label{app:pbcohmeas}

In this part, we provide details and proofs concerning POVM-based coherence measures that are introduced in the main text as the first constituent of our resource theory.

First, we focus on the canonical Naimark extension defined on $\H\ot\H_R$ and a general class of block-coherence measures.

\begin{proposition}\label{prop:schrodinger}
Let $C(\rho,\E)=C(\Ec[\rho],\P)$ be a POVM-based coherence measure, evaluated on the canonical Naimark extension $\P$ of $\E$. The measure can also be expressed as
\begin{align}
C(\rho,\E) = C(\Ec_V[\rho],\{\1\ot\proj{i}\}),
\end{align}
where now the interaction $V$ is attributed to the embedding $\Ec_V[\rho] =V\rho\ot\proj{1}V\ad= \sum_{i,j}A_i\rho A_j\ad\ot\ketbra{i}{j}$.
The measure is invariant under a change of measurement operators, i.e., under the transformation $A_i\to U_iA_i$ with unitary $U_i$. 
\end{proposition}

\begin{proof} 
By definition of $C(\rho,\E)$, $C(\rho',\P)$ is a unitarily-invariant block-coherence measure, that is, $C(\rho',\P)=C(U\rho' U\ad,U\P U\ad)$ holds for all unitaries $U$ on $\H'$ and all states $\rho'\in\Sc'$. Therefore, it holds that
\begin{align}
C(\rho,\E)&=C(\Ec[\rho],\P) \ncl
&=C(\rho\ot\proj{1},\{V\ad\1\ot\proj{i}V\}) \ncl
&=C(V\rho\ot\proj{1}V\ad,\{\1\ot\proj{i}\}) \ncl
&=C(\Ec_V[\rho],\{\1\ot\proj{i}\}),
\end{align}
where the third equality follows from unitary invariance.

This implies that $C(\rho,\E)$ is invariant under a change of measurement operators, $A_i\to U_iA_i$, with unitary $U_i$, and is therefore well-defined. Indeed, the unitary transformation acts on the embedded state as $\Ec_V[\rho]\to\sum_{i,j}U_iA_i\rho A_j\ad U_j\ad\ot\ketbra{i}{j}=U\Ec_V[\rho]U\ad$, with a block-diagonal unitary $U=\sum_iU_i\ot\proj{i}$. Since every block-coherence measure is invariant under block-diagonal unitaries, as shown in Prop.~\ref{prop:buinv}, the measure is invariant under a change of measurement operators.
\end{proof}

In the remainder of this part we prove Lemma \ref{lem:pbmeas} from the main text.
For the claim that $C_{\rel}(\rho,\E)$ is independent of the chosen Naimark extension of a POVM $\E=\{E_i\}$, we consider any Naimark extension $\P$ of $\E$ as defined in Eq.~(\ref{eq:naimark}), not necessarily of tensor product form. To do so, we employ the generalized definitions introduced in App.~\ref{app:naimark},
\begin{align}
\Ec[X]=X\op0,\quad \He=\{\ket{\psi}\op0\},\quad \Pie=\1\op0.
\end{align}
Moreover, we define the generalized POVM-based coherence measure $C_{\rel}(\rho,\E)$ as
\begin{align}
C_{\rel}(\rho,\E)\defeq S(\Delta[\rho\op0])-S(\rho).
\end{align}

\begin{proof}[$\bullet$ Proof of \emph{\bf Lemma 1} from main text]
First, we show that the POVM-based coherence measure obtained from the relative entropy $C_{\rel}(\rho,\E)=C_{\rel}(\rho\op0,\P)$ is independent of the choice of Naimark extension $\P$.
We do that by showing that the eigenvalues of $\Delta[\rho\op0]$ are the same for any two Naimark extensions used to define the dephasing $\Delta$.
The assertion then readily follows because the von Neumann entropy is a function of the eigenvalues of a state.

Let $\P$ on $\H'$ and $\tilde \P$ on $\tilde\H'$ be two Naimark extensions of the same POVM $\E$, such that without loss of generality $d'\geq \tilde d'$ holds. We embed the smaller Hilbert space $\tilde\H'$ canonically into the larger Hilbert space $\H'$ such that all operators on the smaller space are filled up appropriately with zeros.
First, we show that $P_i\rho\op0P_i$ and $\tilde P_i\rho\op0 \tilde P_i$ have the same eigenvalues. 
By definition of the Naimark extension it holds that 
\begin{align}\label{2naimark}
\tr[P_i\rho\op0]=\tr[\tilde P_i\rho\op0]
\end{align}
for all system states $\rho\in\Sc$.
If $\rho$ is a pure state, $P_i\rho\op0P_i$ and $\tilde P_i\rho\op0\tilde P_i$ are both rank-1 operators that because of Eq.~(\ref{2naimark}) have the same nonzero eigenvalue.
For the mixed state case, we consider the following.
From the definition of the Naimark extension it follows that 
\begin{align}\label{aux2}
\Pie P_i P_i\Pie=E_i\op0=\Pie \tilde P_i \tilde P_i \Pie,
\end{align}
where $\tilde P_i$ is extended to $\H'$, implying that $\sum_i\tilde P_i$ is the projector onto $\tilde\H'$. The equation follows from Eq.~(\ref{2naimark}) because the system states provide a POVM-tomography on the subspace $\He$. It is known, that Eq.~(\ref{aux2}) implies that there exists a unitary $Q_i$ on $\H'$ such that 
\begin{align}\label{aux3}
P_i\Pie=Q_i\tilde P_i\Pie.
\end{align}
Concretely, these matrices have singular value decomposition
\begin{align}
P_i\Pie=U_i\begin{pmatrix} \Sigma_{r_i}&\\&0 \end{pmatrix} V_i\ad,\quad
\tilde P_i\Pie=\tilde U_i\begin{pmatrix} \Sigma_{r_i}&\\&0 \end{pmatrix}  V_i\ad,
\end{align}
for some unitaries $U_i,\tilde U_i,V_i$, $r_i=\operatorname{rank}E_i$, and a $r_i\times r_i$ diagonal matrix $\Sigma_{r_i}$ containing the square root of the nonzero eigenvalues of $E_i$.
Then, the unitary is given by $Q_i=U_i \tilde U_i\ad$.
Now, the unitaries $Q_i$ can be combined into a single unitary, by noting that the restriction $Q_i|_{\tilde\pi_i}$ with $\tilde\pi_i=\im\tilde P_i$ is a map from and to orthogonal subspaces $Q_i|_{\tilde\pi_i}\colon\tilde\pi_i\to\pi_i$. Thus, we can define the block-diagonal unitary $Q=\op_iQ_i|_{\tilde\pi_i}\op\1$, where the last term is the identity on the subspace $(\tilde\H')^\perp$ of $\H'$. With that, we have constructed a unitary $Q$ that relates the two Naimark extension acting on the subspace $\He$, namely $P_i\Pie=Q\tilde P_i\Pie$.
Consequently, $\Delta[\rho\op0]=\sum_iP_i\Pie(\rho\op0)\Pie P_i
=Q\sum_i\tilde P_i\Pie(\rho\op0)\Pie\tilde P_iQ\ad=Q\tilde\Delta[\rho\op0]Q\ad$ holds
, i.e., $\Delta[\rho\op0]$ and $\tilde\Delta[\rho\op0]$ have the same eigenvalues.
Since the von Neumann entropy solely depends on the eigenvalues of its argument, we conclude that
\begin{align}
C_{\rel}(\rho\op0,\P)&=S(\Delta[\rho\op0])-S(\rho) \ncl
&=S(\tilde\Delta[\rho\op0])-S(\rho) \ncl
&=C_{\rel}(\rho\op0,\tilde\P),
\end{align}
which means that $C_{\rel}(\rho,\E)$ is independent of the Naimark extension used to define it.

The relative entropy of POVM-based coherence admits an expression just in term of system degrees of freedom, i.e., without making reference to the Naimark space. We need to show that $C_{\rel}$ can be expressed as $C_{\rel}(\rho,\E)=H(\{p_i(\rho)\})+\sum_ip_i(\rho)S(\rho_i)-S(\rho)$, with $p_i(\rho)=\tr[E_i\rho]$, and $\rho_i=\frac1{p_i}A_i\rho A_i\ad$, and where $S$ denotes the von-Neumann entropy, and $H$ the Shannon entropy.
Let $\Delta[\cdot]=\sum_iP_i\cdot P_i$ be the block-dephasing operator of the canonical Naimark extension $\P=\{V\ad\1\ot\proj{i}V\}$ with $V(\rho\ot\proj{1})V\ad=\sum_{i,j}A_i\rho A_j\ad\ot\ketbra{i}{j}$.
Then it holds that
\begin{align}
C_{\rel}(\rho,\E)&=C_{\rel}(\rho\ot\proj{1},\P) \ncl
&=S(\Delta[\rho\ot\proj{1}])-S(\rho\ot\proj{1}) \ncl
&=S(\sum_i\1\ot\proj{i}V(\rho\ot\proj{1})V\ad\1\ot\proj{i})-S(\rho) \ncl
&=S(\sum_iA_i\rho A_i\ad\ot\proj{i})-S(\rho) \ncl
&=S(\sum_ip_i\rho_i\ot\proj{i})-S(\rho) \ncl
&=H(\{p_i(\rho)\})+\sum_ip_iS(\rho_i)-S(\rho),
\end{align}
where the last equality follows from the joint entropy theorem~\cite{nielsenchuang}.
\end{proof}

\section{Minimal and maximal POVM-based coherence}\label{app:minmaxcoh}

In this part, we prove the characterization of POVM-incoherent states from Lemma~\ref{lem:sysincoherent}. Moreover, we show general upper and lower bounds on $C_{\rel}(\rho,\E)$ and discuss classes of POVMs for which these bounds can or cannot be attained.

\begin{proof}[$\bullet$ Proof of \emph{\bf Lemma 2} from main text]
We need to show that $C_{\rel}(\rho,\E)=C_{\rel}(\rho\ot\proj{1},\P)=0$ is equivalent to $E_i \rho E_j=0 \ \forall\ i\neq j\in\{1,\dotsc,n\}$.
The set $\Ic$ of block-incoherent states with respect to the canonical Naimark extension $\P=\{V\ad\1\ot\proj{i}V\}$ is composed of states of the form
\begin{align}
\Ic=\{ V\ad\sum_ip_i\rho_i\ot\proj{i}V \},
\end{align}
as these are the states that are invariant under the dephasing operation $\Delta[\cdot]=\sum_iP_i\cdot P_i$. Here, $\{\rho_i\}$ is a set of states and $\{p_i\}$ a probability distribution. A state $\rho\ot\proj{1}\in\Se$ is of the above form if and only if 
$V\rho\ot\proj{1}V\ad=\sum_ip_i\rho_i\ot\proj{i}$, which is equivalent to
\begin{align}\label{aux6}
&\sum_{i,j}A_i\rho A_j\ad\ot\ketbra{i}{j}=\sum_ip_i\rho_i\ot\proj{i}\ncl
\Leftrightarrow\ \ &A_i \rho A_j\ad=0 \ \ \forall\ i\neq j\in\{1,\dotsc,n\}.
\end{align}
Since $E_i=A_i\ad A_i$, the condition ($\ref{aux6}$) implies $E_i \rho E_j=0 \ \forall\ i\neq j\in\{1,\dotsc,n\}$. The converse implication is also true which can be seen by employing the Moore-Penrose inverse $X\ginv$ of a matrix $X$~\cite{barata2012moore}. It has the properties $X\ginv X=\Pi_{\supp X}$, and $XX\ginv=\Pi_{\im X}$, with the projectors onto the support and image of $X$, respectively. Together with $\supp X\ad=\im X$ and $\supp X\ad=\im X$ it follows that 
\begin{align}
&E_i \rho E_j=0  \ncl
\Leftrightarrow\ &A_i\ad A_i \rho A_j\ad A_j=0  \ncl
\Rightarrow\ &(A_i\ad)\ginv A_i\ad A_i \rho A_j\ad A_jA_j\ginv=0 \ncl
\Leftrightarrow\ &\Pi_{\im A_i} A_i \rho A_j\ad \Pi_{\supp A_j\ad}=0 \ncl
\Leftrightarrow\ &A_i \rho A_j\ad =0.
\end{align}
Since $C_{\rel}(\rho,\E)$ is independent of the choice of Naimark extension, we thus obtain a general characterization of POVM-incoherent states.
\end{proof}

POVM-incoherent states do not exist for any POVM, in particular a certain class of POVMs yields strictly positive coherence for any state. This shows that there can be a finite gap between the set of embedded states $\Se\subseteq\Sc'$ and the set of block-incoherent states $\Ic\subseteq\Sc'$.

\begin{proposition}[Existence of POVM-incoherent states]
Let $\E$ be a POVM whose effects have rank one and no effect is a projector. The set of POVM-incoherent states of $\E$ is empty.
\end{proposition}
\begin{proof}
The system state $\rho$ is incoherent if and only if
$C_{\rel}(\rho,\E)=0$. For rank-1 POVMs, $E_i=\proj{\phi_i}$, the relative entropy of POVM-based coherence is given by
\begin{align}
C_{\rel}(\rho,\E) = H(\{p_i(\rho)\})-S(\rho),
\end{align}
with $p_i(\rho)=\bra{\phi_i}\rho\ket{\phi_i}$. Since no effect is a projector it holds that $\braket{\phi_i}{\phi_i}<1$ for all $i\in\{1,\dotsc,n\}$.
We show that $H(\{p_i(\rho)\})>S(\rho)$ for all states $\rho$. Let $\rho=\sum_k\lambda_k\proj{k}$ be the state's spectral decomposition and let $\eta(p)=-p\log p$ for $p\in[0,1]$, which is a strictly concave function. Then,
\begin{align}
p_i&=\bra{\phi_i}\rho\ket{\phi_i}
=\sum_k\lambda_k \lvert\braket{\phi_i}{k}\rvert^2 
=\sum_k\lambda_k p(i|k)
\end{align}
where $p(i|k)=\lvert\braket{\phi_i}{k}\rvert^2$ is a conditional probability distribution, i.e. 
$\sum_i p(i|k)=1\ \forall k$, while $p(k|i)$ is not since
\begin{align}
\sum_k p(k|i)=\sum_k\lvert\braket{\phi_i}{k}\rvert^2=\braket{\phi_i}{\phi_i}<1.
\end{align}
But $p(k|i)$ can be made a conditional probability distribution by appending an outcome $k_0$ such that
$p(k_0|i)=1-\braket{\phi_i}{\phi_i}$. 
We also set $\lambda_{k_0}=0$ and write $k'=(k_0,k)$ whenever we wish to sum over all outcomes.
With that we can write
\begin{align}
p_i=\sum_k\lambda_k p(i|k)=\sum_{k'}\lambda_{k'} p(k'|i).
\end{align}
Now, we can derive a relation between $H(\{p_i\})$ and $S(\rho)$,
\begin{align}
H(\{p_i\})&=\sum_i\eta(p_i) = \sum_i\eta(\sum_{k'} p(k'|i) \lambda_{k'})	\ncl
&\geq \sum_{i,k'} p(k'|i)\eta(\lambda_{k'}) = \sum_{i,k} p(i|k)\eta(\lambda_{k}) \ncl
&=\sum_k\eta(\lambda_k)= S(\rho),
\end{align}
where the inequality follows from the concativity of $\eta$, and we have used that 
$\eta(\lambda_{k_0})=0$ and $p(k|i)=p(i|k)$.
Since $\eta$ is a strictly concave function, equality in the third row is equivalent to one of the following conditions
\begin{enumerate}[label=\roman{*}.]
\item $p(k|i)=\delta_{k=f(i), i} \quad \forall k\ \textrm{with}\ \lambda_{k}\neq0$,
\item $\lambda_{k'}=\lambda_{l'}$ for all $k',l'$ with $p(k'=l'|i)\neq\delta_{k'=f(i), i}$
\end{enumerate}
where $f$ denotes an index function.
The second condition is formulated as above because for the qutrit incoherent basis $\{\ket{i}\}_{i=1}^3$ the state $\frac14(\proj{+}+\proj{-})+\frac12\proj{2}$ with $\ket{\pm}=\frac1{\sqrt2}(\ket{0}\pm\ket{1})$ is incoherent.
For the POVM under consideration, the first condition cannot be met since $|\braket{\psi}{\phi_i}|^2<1$ for all $i$ and all states $\ket{\psi}$.
The second condition also cannot be met because $p(k'|i)\neq\delta_{k'=f(i), i}$ holds for all
$k' \in\{0,\dotsc,d\}$. But not all eigenvalues can be equal since $\sum_{k'}\lambda_{k'}=1$ while $\lambda_{k_0}=0$. Moreover, the second condition cannot be fulfilled for any POVM with $\tr[E_i]<1$ for some index $i$, since then the additional outcome $k_0$ with $\lambda_{k_0}=0$ is needed.
This implies that the maximally mixed state is not incoherent for such measurements.
Altogether we conclude that no system incoherent states exist for the considered POVM.
\end{proof}

Since system incoherent states do not exist for any measurement, it is important to characterize states with minimal and maximal POVM-based coherence. 
The measure $C_{\rel}(\rho,\E)$ is bounded by the extremal values of the corresponding block-coherence measure on $\Sc'$ given by $0\leq C(\rho',\P)\leq \log(d')$. However, the upper bound can be made tighter.

\begin{proposition}
Let $\E$ be an $n$-outcome POVM. The POVM-based coherence measure $C_{\rel}(\rho,\E)$ satisfies the bounds $0\leq C_{\rel}(\rho,\E)\leq\log n$.
\end{proposition}
\begin{proof}
We show the upper bound. First, we consider the pure state case $\rho=\proj{\psi}$, for which the measure reads
\begin{align}
C_{\rel}(\ket{\psi},\E) = H(\{p_i(\ket{\psi})\}).
\end{align}
The expression is maximized for states with uniform outcomes $p_i=\frac1n$ which yields $H(\{p_i(\ket{\psi})\}) \leq \log n$.
Since $C_{\rel}(\rho,\E)$ is a convex function, i.e., it decreases under mixing, the maxima are attained by pure states, and thus $C_{\rel}(\rho,\E)\leq \log n$ also holds for mixed states.
\end{proof}

The convexity of $C_{\rel}$ implies that the maximum coherence of a POVM is attained by the pure states with highest outcome entropy. However, analytically maximizing $C_{\rel}$ even for pure states is generally hard, see e.g., Ref.~\cite{szymusiak2017pure}, where the maximal value for informationally complete symmetric qubit POVMs was obtained.
Examples for POVMs that attain the upper bound are the qubit trine POVM, but also informationally complete POVMs, namely those for which there are pure states with maximal randomness gain~\cite{bischof2017measurement}. Moreover, one can readily construct rank-one POVMs in any dimension that achieve the upper bound.

Finally, we discuss states which minimize $C_{\rel}$ for a given POVM.
Because $C_{\rel}$ is a convex function on a convex set it can be shown that the set $\Mc$ of its minima is convex. 
In the qubit case, the states with minimal coherence can be found analytically. Qubit quantum states can be parameterized as $\rho(\vec r)=\frac12(\1+\vec r\cdot\vec\sigma)$ with Bloch vector $\modu{\vec r}\leq1$, and $\vec r\cdot\vec\sigma=\sum_ir_i\sigma_i$, where $\sigma_i$ denotes the $i$-th Pauli matrix. The function $\rho(\vec r)$ is affine in $\vec r$ and thus $C_{\rel}(\vec r)\defeq C_{\rel}(\rho(\vec r))$ is convex. Consequently, for any fixed POVM $\E$ we have the following optimization problem
\begin{align}
\trm{minimize}&\quad C_{\rel}(\vec r) \ncl
\trm{such that}&\quad \modu{\vec r}^2-1\leq 0
\end{align}
This is a convex optimization problems, i.e., the objective function $C_{\rel}(\vec r)$ and the inequality contraint function $g(\vec r)=\modu{\vec r}^2-1$ are convex. For such problems it is known that  any point $\vec r^*$ that fulfills the Karush–Kuhn–Tucker (KKT)~\cite{boyd2004convex} conditions is a global minimum of the objective function. One can readily check that for the problem above a point $\vec r^*$ fulfills the KKT conditions if  
\begin{align}\label{eq:kkt}
\modu{\vec r^*}^2\leq1 \quad\textrm{and}\quad \nabla_{\vec r}C_{\rel}(\vec r^*)=0.
\end{align}
Therefore, given a POVM $\E$, the minimum of $C_{\rel}(\rho,\E)$ is achieved for states $\rho(\vec r^*)$ with $\vec r^*$ from Eq.~(\ref{eq:kkt}). In dimensions higher than two, a similar analysis can be carried out with more involved constraints.

\section{POVM-incoherent operations}\label{app:povmicops}

In this part, we provide proofs for the general results concerning POVM-incoherent operations from the main text. In particular, we present the semidefinite programs that characterize the set of POVM-incoherent operations and the fidelity $F_{\max}(\rho,\sigma)$, respectively.

\begin{proof}[$\bullet$ Proof of \emph{\bf Lemma 3} from main text]
Let $\Lambda_\pic$ be a POVM-incoherent operation with respect to the POVM $\E$. By definition there exists a channel $\Lambda'_\bic$ on $\Sc'$ obeying the two properties from Def.~\ref{def:povmmio} such that
$\Lambda_\pic[\rho]\op0 = \Lambda'_\bic[\rho\op0]$.
Thus, it holds that 
\begin{align}
C(\Lambda_\pic[\rho],\E)&= C(\Lambda_\pic[\rho]\op0,\P) \ncl
&= C(\Lambda'_\bic[\rho\op0],\P) \ncl
&\leq C(\rho\op0,\P)= C(\rho,\E),
\end{align}
where the inequality is a consequence of $\Lambda'_\bic$ being an block-incoherent operation with respect to $\P$.
\end{proof}

If the Naimark space has tensor product form $\H\ot\H_R$, then due to subspace-preservation $\Lambda'$ can be decomposed as
\begin{align}
\Lambda' &=\Omega\circ\Lambda'\circ\Omega + \Lambda'\circ\Omega^\perp \ncl
&=(\Lambda\ot\1)\circ\Omega + \Lambda'\circ\Omega^\perp,
\end{align}
where $\Lambda$ is a channel on $\Sc$, $\Omega[\rho']=\Pie\rho'\Pie$ and $\Omega^\perp=\id-\Omega$. Thus, in this case we have $\Lambda'|_{\Se}=\Lambda\ot\1$, leading to the local operation $\Lambda$ on $\Sc$.

In the following, we show that the set of POVM-incoherent operations of a POVM $\E$ is independent of the choice of Naimark extension used for its definition. We consider any Naimark extension $\P$ of $\E$ as defined in Eq.~(\ref{eq:naimark}), not necessarily of tensor product form.
For that, it is instructive to read the proof of Lemma~\ref{lem:pbmeas} established in App.~\ref{app:pbcohmeas}. There, we summarized the generalized embedding definitions introduced in App.~\ref{app:naimark},
\begin{align}
\Ec[X]=X\op0,\quad \He=\{\ket{\psi}\op0\},\quad \Pie=\1\op0.
\end{align}

\begin{proof}[$\bullet$ Proof of \emph{\bf Theorem 1} from main text]
First, we prove that the set of POVM-incoherent operations is independent of the choice of Naimark extension used for its definition. 
Let $\P$ on $\H'$ and $\tilde \P$ on $\tilde\H'$ be two general Naimark extensions of the same POVM $\E$, such that without loss of generality $d'\geq \tilde d'$ holds. 
We embed the smaller Hilbert space $\tilde\H'$ canonically into the larger Hilbert space $\H'$, and all operators on the first space are filled up appropriately with zeros.
Let $\Lambda'_\P$ ($\Lambda'_{\tilde\P}$) be a channel on $\Sc'$ ($\tilde\Sc'$) obeying the two properties from Def.~\ref{def:povmmio}. Since $d'\geq \tilde d'$, the proof idea is to show that for every $\Lambda'_\P$ there exists a $\Lambda'_{\tilde\P}$ such that 
$\Lambda_\pic\defeq\Ec\ad\circ\Lambda'_\P\circ\Ec=\Ec\ad\circ\Lambda'_{\tilde\P}\circ\Ec$, i.e., the two channels lead to the same POVM-incoherent operation.

For that, we parameterize the BIC property, i.e., the first property of Def.~\ref{def:povmmio}. Let $\Lambda'$ be a superoperator on $\Sc'$, 
which can be made block-incoherent in the following way.
If the linear map
\begin{align}\label{aux7}
\Lambda'_\P&=\Lambda'-\Delta^\perp\circ\Lambda'\circ\Delta \ncl
&=\Lambda'+\Delta\circ\Lambda'\circ\Delta-\Lambda'\circ\Delta,
\end{align}
is a channel, it is in the class MIO, since then and only then $\Lambda'_\P\circ\Delta=\Delta\circ\Lambda'_\P\circ\Delta$ holds, which is the definition of a block-incoherent channel. Here, $\Delta^\perp=\id-\Delta$ is the superoperator that sets the blocks on the diagonal to zero. Hence, we substract the part of $\Lambda'$ which can create block coherence, namely $\Delta^\perp\circ\Lambda'\circ\Delta$. 

In the proof of Lemma \ref{lem:pbmeas}, we have established the identity $P_i\Pie=Q\tilde P_i\Pie$ with a unitary operator $Q$. This implies that
$\Delta[\rho\op0]=Q\tilde\Delta[\rho\op0]Q\ad$, which in turn yields $\Delta\circ\Ec=\Qc\circ\tilde\Delta\circ\Ec$, with the unitary channel $\Qc[\cdot]=Q\cdot Q\ad$.
By taking the adjoint we obtain $\Ec\ad\circ\Delta=\Ec\ad\circ\tilde\Delta\circ\Qc\ad$. 
Another implication of $P_i\Pie=Q\tilde P_i\Pie$ is obtained by summing over $i$, leading to
$Q\Pie=\Pie$, which in turn implies that $\Qc|_{\Se}=\id$, and also $\Qc\circ\Ec=\Ec$, because $\im\Ec=\Se$.

Finally, we can investigate the relation of two POVM-incoherent operations defined with respect to $\P$ and $\tilde\P$, respectively, by employing Eq.~(\ref{aux7}),
\begin{align}\label{aux8}
&\ \Lambda_{\P}=\Ec\ad\circ\Lambda'_\P\circ\Ec \ncl
&=\Ec\ad\circ\Lambda'\circ\Ec + \Ec\ad\circ\Delta\circ\Lambda'\circ\Delta\circ\Ec
-\Ec\ad\circ\Lambda'\circ\Delta\circ\Ec \ncl
&=\Ec\ad\circ\Lambda'\circ\Ec + \Ec\ad\circ\tilde\Delta\circ\Qc\ad\circ\Lambda'\circ\Qc\circ\tilde\Delta\circ\Ec
-\Ec\ad\circ\Lambda'\circ\Qc\circ\tilde\Delta\circ\Ec \ncl
&=\Ec\ad\circ\Qc\ad\circ\Lambda'\circ\Qc\circ\Ec + \Ec\ad\circ\tilde\Delta\circ\Qc\ad\circ\Lambda'\circ\Qc\circ\tilde\Delta\circ\Ec \ncl
&\quad-\Ec\ad\circ\Qc\ad\circ\Lambda'\circ\Qc\circ\tilde\Delta\circ\Ec \ncl
&=\Ec\ad\circ\tilde\Lambda'\circ\Ec + \Ec\ad\circ\tilde\Delta\circ\tilde\Lambda'\circ\tilde\Delta\circ\Ec-\Ec\ad\circ\tilde\Lambda'\circ\tilde\Delta\circ\Ec \ncl
&=\Ec\ad\circ\Lambda'_{\tilde \P}\circ\Ec=\Lambda_{\tilde\P},
\end{align}
where in the third and fourth equality we have substituted the relations $\Delta\circ\Ec=\Qc\circ\tilde\Delta\circ\Ec$ and $\Qc\circ\Ec=\Ec$ (and its adjoint), respectively. 
Here,  $\tilde \Lambda'\defeq\Qc\ad\circ\Lambda'\circ\Qc$ and
$\Lambda'_{\tilde\P}=\tilde\Lambda'+\tilde\Delta\circ\tilde\Lambda'\circ\tilde\Delta-\tilde\Lambda'\circ\tilde\Delta$ are channels. 
However, $\Lambda'_{\tilde\P}$ cannot be considered to be in $\mio$ with respect to $\tilde \P$, because it can possibly map out of $\tilde\Sc'$. We can solve that issue as follows. Note that the restriction of $\tilde Q\defeq Q|_{\tilde\H'}\to\H'$ to $\tilde\H'$ is an isometry, from which we obtain the isometric channel $\tilde{\Qc}[\cdot]=\tilde Q\cdot \tilde Q$.
For an isometric channel we can define a reversal channel $\mathcal R[\rho']\defeq\tilde{\Qc}\ad[\rho']+\tr[(\1-\Pi_{\im\Qc})\rho']\tilde\sigma$, where $\rho'\in\Sc'$ and $\tilde\sigma$ is a state in $\tilde\Sc'$ such that $\Ec\ad[\tilde\sigma]=\Ec\ad\circ\Delta[\tilde\sigma]=0$. Now define $\bar\Lambda'=\mathcal R\circ\Lambda'\circ\tilde{\Qc}$ and $\bar\Lambda'_{\tilde\P}=\bar\Lambda'+\tilde\Delta\circ\bar\Lambda'\circ\tilde\Delta-\bar\Lambda'\circ\tilde\Delta$, which is in $\mio$ of $\tilde \P$. Moreover, by construction it holds that $\Ec\ad\circ\mathcal R=\Ec\ad\circ\Qc\ad$ and $\Ec\ad\circ\Delta\circ\mathcal R=\Ec\ad\circ\Delta\circ\Qc\ad$, which implies that $\Ec\ad\circ\bar\Lambda'_{\tilde \P}\circ\Ec=\Lambda_{\tilde\P}$.
Thus, Eq.~(\ref{aux8}) implies that $\Lambda'_\P$ and $\bar\Lambda'_{\tilde \P}$ lead to the same POVM-incoherent operation which is the desired relation from which the independence property follows.

With the independence property established in the previous paragraph we can show that if $\E$ is an orthogonal rank-1 measurement POVM-incoherent operations are equivalent to coherence $\mio$ channels. Since in this case $\E$ is already projective, we can choose the trivial Naimark space $\H'=\H\ot\Cb\simeq \H$ and $\P=\E$. Then, subspace-preservation is trivially fulfilled for all channels from $\Sc$ to itself, while the block-incoherent condition is equivalent to the $\mio$ condition in standard coherence theory. Since POVM-incoherent operations are independent of the chosen Naimark extension the assertion also holds for any other Naimark extension of $\E$.

Finally, we show that the set of POVM-incoherent operations can be characterized by a semidefinite program. Let $\mathcal B=\{B_\alpha\}_\alpha=\{\ketbra{i}{j}\}_{i,j=1}^{d'}$ be the (Hilbert-Schmidt-orthonormal) standard matrix basis of operators on $\H'$ in lexicographical order. Let $\vecr:\Sc'\to\Cb^{d'^2}$ be the isomorphism that maps a state $\rho$ on the Naimark space to its coordinate vector $\vecr(\rho)$ with respect to $\mathcal B$.
To any superoperator $\Lambda'$ on $\Sc'$, we associate its coordinate matrix with respect to $\mathcal B$, called the process matrix,
\begin{align}
\hat \Lambda'_{\alpha,\beta} = \tr[B_\alpha\ad\Lambda'[B_\beta]],
\end{align}
which has the property that $\hat\Lambda' \vecr(\rho)=\vecr(\Lambda'[\rho])$~\cite{havel2003robust}. 
The process matrix $\hat \Lambda'$ is related to the Choi matrix $J(\Lambda')$ of $\Lambda'$ as~\cite{havel2003robust}
\begin{align}
\hat \Lambda' = d'J(\Lambda')^R, \quad X^R\defeq \sum_\alpha(\1\ot B_\alpha)X(B_\alpha\ot\1),
\end{align}
where the mapping $X\to X^R$ is an involution, called row-reshuffling~\cite{wood2011tensor}.
On the level of transfer matrices the composition of superoperators $\Ec\circ\Fc$ becomes multiplication $\hat\Ec\hat\Fc$.
With that we can characterize POVM incoherent operations via a semidefinite feasibility problem. A system channel $\Lambda$ on $\Sc$ is POVM-incoherent if and only if there exists a Choi matrix $J$ on $\H'\ot\H'$ such that 
\begin{align}
\textrm{find:}\quad&\hat{\Ec}\ad J^R \hat{\Ec}=\hat\Lambda \ncl
\textrm{subj. to:}\quad&J\geq 0, \quad \tr_1J=\textstyle{\frac{\1}{d'}}, \ncl
&J^R\hat\Delta = \hat\Delta J^R \hat\Delta, \ncl
&J^R\hat\Omega = \hat\Omega J^R \hat\Omega.
\end{align}
Here $\tr_1$ denotes the trace over the first subsystem of $\H'\ot\H'$,
and $\Ec[\rho]=\rho\op0$. Moreover, $\Delta$ denotes the block-dephasing operator and $\Omega[\rho']\defeq\Pie\rho'\Pie$, with $\Pie$ being the projector onto $\Se$.
The SDP characterization allows for an efficient numerical check whether a channel is element of the set of POVM-incoherent operations.
\end{proof}

The fidelity between two quantum states $\rho,\sigma$ is given by
$F(\rho,\sigma)=\tr\sqrt{\sqrt\rho\sigma\sqrt\rho}$. We define the quantity $F_{\max}(\rho,\sigma)=\max_{\Lambda_\pic}F(\Lambda_\pic[\rho],\sigma)$ between the states $\sigma$ and $\Lambda_\pic[\rho]$, maximized over all POVM-incoherent operations $\Lambda_\pic$ of a POVM $\E$. The quantity characterizes the usefulness of a particular state $\rho$ when only POVM-incoherent operations can be implemented, as it provides a measure of how well $\sigma$ can be approximated.
As a consequence of the SDP characterization of POVM-incoherent operations we are able to efficiently numerically calculate $F_{\max}$.

\begin{proposition}
The fidelity $F_{\max}(\rho,\sigma)=\max F(\Lambda_\pic[\rho],\sigma)$ equals the solution of the following semidefinite program
\begin{align}
F_{\max}(\rho,\sigma)= \ncl
\textrm{maximize:}\quad&\frac12(\tr[X]+\tr[X\ad]) \ncl
\textrm{subj. to:}\quad&
\begin{pmatrix}
\sigma & X \\
X\ad & \Lambda[\rho]
\end{pmatrix}\geq 0,  \ncl
& \Lambda[\rho]=\vecr\inv(\hat{\Ec}\ad J^R\hat\Ec \vecr(\rho)) \ncl
& J\geq 0, \quad \tr_1J=\textstyle{\frac{\1}{d'}}, \ncl
& J^R\hat\Delta = \hat\Delta J^R \hat\Delta, \ncl
& J^R\hat{\Omega} = \hat{\Omega} J^R \hat{\Omega}.
\end{align}
\end{proposition}
\begin{proof}
The fidelity between two arbitrary quantum states can be cast in the form of an SDP~\cite{piani2016hierarchy,watrous2012simpler}. Combining this with the SDP characterization of POVM-incoherent operations from the proof of Theorem 1 proves the assertion.
\end{proof}

\section{Example: qubit trine POVM}\label{app:qubitcoh}

In this section we apply all previously obtained results to study the POVM-based coherence theory of qubit POVMs.

\subsubsection*{Coherence theory of mixed-unitary channel}
The simplest example for a POVM-based coherence theory is obtained from the Kraus operators of a mixed-unitary channel which lead to the POVM $\E=\{p_1\1,\dotsc,p_{n}\1\}$ with a probability distribution $\{p_i\}$. 
The canonical Naimark extension is given by $\P=\{\1\ot\proj{\varphi_i}\}$, where $\ket{\varphi_i}$ is an orthonormal basis of $\H_R$ such that $\lvert\braket{\varphi_i}{1}\rvert^2=p_i$. 
In this case all states have the same coherence of $C_{\rel}(\rho,\E)=H(\{p_i\})$, since the system-apparatus interaction is just a local unitary generating coherence in the measurement apparatus.
As a consequence all system channels $\Lambda$ on $\Sc$ are POVM-incoherent, since the embedding $\Lambda\ot\id$ is subspace-preserving and commutes with the dephasing operation $\Delta[\cdot]=\sum_iP_i\cdot P_i$ and thus maps incoherent states to themselves, $(\Lambda\ot\id)[\Delta[\rho]]=\Delta[(\Lambda\ot\id)[\rho]]$.

\subsubsection*{Coherence theory of qubit trine POVM}
In this section we apply the results from the main text to study the POVM-based coherence theory of the qubit trine POVM which is given by
\begin{align}\E=
 \left\{ 
\frac13
\begin{pmatrix}
1 & 1 \\
1 & 1 \\
\end{pmatrix},
\frac13
\begin{pmatrix}
1 & \omega^* \\
\omega & 1 \\
\end{pmatrix},
\frac13
\begin{pmatrix}
1 & \omega \\
\omega^* & 1 \\
\end{pmatrix}
 \right\},
\end{align}
with $\omega=e^{\frac{2\pi}3 i}$, and $\omega^*=\omega^2$.
Since our resource theory is independent of the choice of Naimark extension, it is numerically advantageous to employ the Naimark extension of smallest dimension. Such a minimal Naimark extension of $\E$ is given by $\P=\{\proj{\varphi_i}\}_{i=1}^3$ on $\H'=\Cb^3$ with
\begin{align}
\ket{\varphi_1} &=  \frac{1}{\sqrt{3}} (\ket{1}+\ket{2}+\ket{3}) \\
\ket{\varphi_2} &=  \frac{1}{\sqrt{3}} (\ket{1}+\omega\ket{2}+\omega^*\ket{3}) \\
\ket{\varphi_3} &= \frac{1}{\sqrt{3}}  (\ket{1}+\omega^*\ket{2}+\omega\ket{3}.
\end{align}
Any incoherent state on $\Sc'$ with respect to $\P$ can be written as
\small
\begin{align}
&\rho_\inc= \sum_i^3p_i\proj{\varphi_i}= \ncl
&\frac13
\begin{pmatrix}
 1 & {p_1}+\omega^*{p_2}+\omega {p_3} & {p_1}+\omega p_2+\omega^* {p_3} \\
 {p_1}+\omega{p_2}+\omega^* {p_3} & 1 & p_1 + \omega^* p_2 + \omega p_3 \\
 {p_1}+\omega^* p_2+\omega {p_3} & p_1 + \omega p_2 + \omega^* p_3 & 1 \\
\end{pmatrix},
\end{align}
\normalsize
for some probability distribution $\{p_i\}_{i=1}^3$.
Moreover, a general embedded system state in $\Se\subseteq\Sc'$ is of the form
\begin{align}
\Ec[\rho]=\rho\op0 = 
\begin{pmatrix}
 \rho_{11} & \rho_{12} & 0 \\
 \rho_{21} & \rho_{22} & 0 \\
 0 & 0 & 0 \\
\end{pmatrix},
\end{align}
with $\rho\in\Sc$. Finally, any dephased embedded state reads as
\begin{align}
\Delta[\rho\op0] = \frac13
\begin{pmatrix}
1 & {\rho_{12}} & {\rho_{21}} \\
 {\rho_{21}} & 1 & {\rho_{12}} \\
 {\rho_{12}} & {\rho_{21}} & 1 \\
\end{pmatrix}.
\end{align}

We now provide a characterization of the POVM-incoherent unitaries of the trine POVM. Let $U'$ be a unitary on the Naimark space $\H'=\Cb^3$ that is
\begin{enumerate}[label=\roman{*}.]
 \item \emph{(Naimark-) incoherent:} $U'\ket{\varphi_i} \propto \ket{\varphi_j}$
 \item \emph{Subspace-preserving:} $U' \begin{pmatrix}\ket{\psi}\\0\end{pmatrix} = \begin{pmatrix}\ket{\psi'}\\0\end{pmatrix}$,
\end{enumerate}
with $\ket{\psi},\ket{\psi'}\in\Cb^2$ and $\ket{\varphi_i}$ being the i-th measurement vector of the Naimark extension. Then we call the operator $U^{\operatorname{trine}}$ on $\H=\Cb^2$ given by
\begin{align}
U^{\operatorname{trine}}=\begin{pmatrix} 1 & 0 & 0 \\
0 & 1& 0
\end{pmatrix}
U'
\begin{pmatrix} 1 & 0  \\
0 & 1 \\
0 & 0
\end{pmatrix}
\end{align}
a POVM-incoherent unitary of the trine POVM.
It has been shown~\cite{bai2015maximally} that all qutrit incoherent unitaries are of the form
\begin{align}
U_\pi' = \sum_{i=1}^3 e^{i\alpha_i} \ketbra{\varphi_{\pi(i)}}{\varphi_i},
\end{align}
where $\alpha_i\in\Rb$ and $\pi=(\pi(1)\ \pi(2)\ \pi(3))\in S_3$ is one of the six permutations of a three-element set. Thus, there are six classes of 3-parameter incoherent unitaries.
Moreover, the subspace-preservation condition \romannumeral2.\ is fulfilled for all $\ket{\psi}\in\Cb^2$ if and only if the unitary is of the form
\begin{align}\label{aux11}
U_\pi' = 
\begin{pmatrix}
 * & * & * \\
 * & * & * \\
 0 & 0 & * \\
\end{pmatrix},
\end{align}
where $*$ denotes some complex entry. Therefore, to obtain a POVM-incoherent unitary we require that $U_\pi'$ satisfies $(U'_\pi)_{3,1} = (U'_\pi)_{3,2}=0$. This yields the POVM-incoherent unitary $U^{\operatorname{trine}}$ as the upper left $2\times2$ block of the resulting matrix. The following list contains all trine POVM-incoherent unitaries:
\begin{align}U_{(123)}^{\operatorname{trine}}&=
\begin{pmatrix}
1 & 0 \\
0 & 1 \\
\end{pmatrix}
=\1, \ncl
U_{(231)}^{\operatorname{trine}}&=
\begin{pmatrix}
\sqrt{\omega^*} & 0\\
0 & \sqrt\omega \\
               \end{pmatrix} = R_{\vec e_z}(\frac{2\pi}3), \ncl
U_{(312)}^{\operatorname{trine}}&=\begin{pmatrix}
\omega^* & 0 \\
0 & \omega \\
               \end{pmatrix} = R_{\vec e_z}(\frac{4\pi}3),\ncl
U_{(132)}^{\operatorname{trine}}&=
\begin{pmatrix}
0 & -i \\
-i & 0\\	
               \end{pmatrix} = R_{\vec m_1}(\pi), \ncl
U_{(321)}^{\operatorname{trine}}&=
\begin{pmatrix}
0 & \omega^{\frac54} \\
\omega^{\frac14} & 0\\	
               \end{pmatrix} = R_{\vec m_2}(\pi),\ncl
U_{(213)}^{\operatorname{trine}}&=
\begin{pmatrix}
0 & \omega^{\frac14} \\
\omega^{\frac54} & 0 \\
               \end{pmatrix} = R_{\vec m_3}(\pi),
\end{align}
Up to a phase, any qubit unitary can be expressed as $R_{\vec n}(\theta) = e^{-i\frac{\theta}2 \vec n\cdot \vec\sigma}\in\operatorname{SU}(2)$, namely as the rotation around the bloch vector $\vec n$ with angle $\theta$. Here, $\vec m_i$ denotes the Bloch vector of the measurement vector $\ket{\varphi_i}$, and $U_\pi^{\operatorname{trine}}$ denotes the POVM-incoherent unitary obtained from $U'_\pi$.
This set is composed of the six rotations that leave the equilateral triangle, whose vertices are given by the measurement direction vectors $\{\vec m_i\}$, invariant.
There are no continuous degrees of freedom left, since the two subspace-preserving conditions together with the requirement of having unit determinant uniquely determines the parameters $\alpha_i$.

At last, we discuss the usefulness of a maximally coherent state $\ket{\Psi_{\operatorname{m}}}\in\{\ket{0},\ket{1}\}$ for the POVM-based coherence theory of the trine POVM. We have numerical evidence that the transformation $\proj{\Psi_{\operatorname{m}}}\to\rho$ with $\rho\in\Sc$ is always possible by a POVM-incoherent map.
Concretely, by plotting the value of $F_{\max}(\ket{\Psi_{\operatorname{m}}},\sigma)$ for any pure state $\sigma=\proj{\psi}$, we observe that all pure states can be obtained with certainty from $\ket{\Psi_{\operatorname{m}}}$ under POVM-incoherent operations. Therefore, all qubit states $\rho=\sum_ip_i\proj{i}$ can be obtained from $\ket{\Psi_{\operatorname{m}}}$ by a POVM-incoherent map, namely via preparing the eigenstate $\ket{i}$ with probability $p_i$.

\end{document}